\documentclass[a4paper,USenglish,nolineno]{socg-lipics-v2019}

\newcommand{\R}{\ensuremath{\mathbb{R}}}
\newcommand{\po}{\ensuremath{\mathcal{P}}}
\newcommand{\Tr}{\ensuremath{\mathcal{T}}}
\newcommand{\var}{\ensuremath{\mathcal{V}}}
\newcommand{\h}{\ensuremath{\mathcal{H}}}
\newcommand{\ttt}{\mathbf{t}}
\newcommand{\ff}{\mathbf{f}}
\newcommand{\0}{^{\mathbf{0}}}
\newcommand{\1}{^{\mathbf{1}}}

\DeclareMathOperator{\sign}{sign}

\theoremstyle{plain}
\newtheorem{conjecture}[theorem]{Conjecture}
\newtheorem{obs}[theorem]{Observation}

\newcommand{\pth}[1]{\left(#1\right)}
\newcommand{\vc}[1]{\overrightarrow{#1}}
\newcommand*{\eqdef}{\stackrel{\text{\tiny{def}}}{=}}

\definecolor{codegreen}{rgb}{0,0.3,0}
\definecolor{codegray}{rgb}{0.5,0.5,0.5}
\definecolor{codedarkred}{rgb}{0.54,0,0}
\definecolor{backcolour}{rgb}{0.96,0.96,0.94}

\lstdefinestyle{mystyle}{
    backgroundcolor=\color{backcolour},   
    commentstyle=\color{codegreen},
    keywordstyle=\bf\color{codedarkred},
    numberstyle=\tiny\color{codegray},
    stringstyle=\color{codedarkred},
    basicstyle=\footnotesize,
    breakatwhitespace=false,         
    breaklines=true,                 
    captionpos=b,                    
    keepspaces=true,                 
    numbers=left,                    
    numbersep=5pt,                  
    showspaces=false,                
    showstringspaces=false,
    showtabs=false,                  
    tabsize=2
}
 
\title{An Experimental Study of Forbidden Patterns in Geometric Permutations by Combinatorial Lifting}
\titlerunning{Forbidden patterns in geometric permutations by combinatorial lifting}

\author{Xavier Goaoc}{Universit\'e de Lorraine, CNRS, Inria, LORIA, F-54000 Nancy, France}{xavier.goaoc@loria.fr}{}{Supported by Institut Universitaire de France.}
\author{Andreas Holmsen}{Department of Mathematical Sciences, KAIST, Daejeon, South Korea}{andreash@kaist.edu}{}{Supported by Basic Science Research Program through the National Research Foundation of Korea (NRF) funded by the Ministry of Education (NRF-2016R1D1A1B03930998).}
\author{Cyril Nicaud}{Universit\'e Paris-Est, LIGM (UMR 8049), CNRS, ENPC, ESIEE, UPEM, F-77454, Marne-la-Vall\'ee, France.}{cyril.nicaud@u-pem.fr}{}{}

\authorrunning{Goaoc X., Holmsen A., and Nicaud C.}

\Copyright{Xavier Goaoc, Andreas Holmsen and Cyril Nicaud}
\ccsdesc{Theory of computation $\rightarrow$ Randomness}
\ccsdesc{Geometry and discrete structures $\rightarrow$ Computational geometry}
\ccsdesc{Computing methodologies $\rightarrow$ Symbolic and algebraic manipulation $\rightarrow$ Symbolic and algebraic algorithms $\rightarrow$ Combinatorial algorithms}
\keywords{Geometric permutation; Emptiness testing of semi-algebraic sets; Computer-aided proof}
\category{}

\hideLIPIcs

\begin{document}

\maketitle

\begin{abstract}
  We study the problem of deciding if a given triple of permutations
  can be realized as geometric permutations of disjoint convex sets in
  $\mathbb{R}^3$. We show that this question, which is equivalent to deciding
  the emptiness of certain semi-algebraic sets bounded by cubic
  polynomials, can be ``lifted'' to a purely combinatorial problem. We
  propose an effective algorithm for that problem, and use it to gain
  new insights into the structure of geometric permutations.
\end{abstract}

\section{Introduction}

Consider pairwise disjoint convex sets $C_1, C_2, \ldots, C_n$ and
lines $\ell_1, \ell_2, \ldots, \ell_k$ in~$\R^d$, where every line
intersects every set. Each line $\ell_i$ defines two orders on the
sets, namely the orders in which the two orientations of $\ell_i$ meet
the sets; this pair of orders, one the reverse of the other, are
identified to form the \emph{geometric permutation} realized by
$\ell_i$ on $C_1, C_2, \ldots, C_n$. Going in the other direction, one
may ask if a given family of permutations can occur as geometric
permutations of a family of pairwise disjoint convex sets in $\R^d$,
\emph{i.e.}  whether it is \emph{geometrically realizable} in~$\R^d$.

In $\R^2$ there exist pairs of permutations that are unrealizable,
while in $\R^3$, every pair of permutations is realizable by a family
of segments with endpoints on two skew lines. The simplest non-trivial
question is therefore to understand which triples of permutations are
geometrically realizable. This question is equivalent to testing the
non-emptiness of certain semi-algebraic sets bounded by cubic
polynomials. We show that the structure of these polynomials allow to
"lift" this algebraic question to a purely combinatorial one, then
propose an algorithm for that combinatorial problem, and present some
new results on geometric permutations obtained with its assistance.

\subparagraph{Conventions.}

To simplify the discussion, we work with \emph{oriented} lines and
thus with \emph{permutations}, in place of the non-oriented lines and
geometric permutations customary in this line of inquiry.  We
represent permutations by words such as $1423$ or $badc$, to be
interpreted as follows. The letters of the word are the elements being
permuted and they come with a natural order, namely $<$ for integer
and the alphabetical order for letters. The word gives the sequence of
images of the elements by increasing order; for example, $312$ codes
the permutation mapping $1$ to $3$, $2$ to $1$ and $3$ to $2$, and
$badc$ codes the permutation exchanging $a$ with $b$ and $c$
with~$d$. The \emph{size} of a permutation is the number of elements
being permuted, that is the length of this word. We say that a triple
of permutations is \emph{realizable} (resp. \emph{forbidden}) to mean
that it is realizable (resp. not realizable) in $\R^3$.

\subsection{Contributions}\label{s:results}

Our results are of two types, methodological and geometrical.

\subparagraph{Combinatorial lifting.}

Our first contribution is a new approach for deciding the emptiness of
a semi-algebraic set with a special structure. We describe it for the
geometric realizability problem, here and in
Section~\ref{s:combilift}, but stress that it applies more broadly.

\bigskip

As spelled out in Section~\ref{s:parameter}, deciding if a triple of
permutations is realizable amounts to testing the emptiness of a
semi-algebraic set $R \subseteq \R^n$. Let $u_1, u_2, \ldots, u_n$
denote the variables and $P_1, P_2, \ldots, P_m$ the polynomials used
in a Boolean formula defining $R$. The structure we take advantage
of is that here, each $P_k$ can be written as a product of terms, each
of which is of the form $u_i-u_j$, $u_i-1$, or $u_i-f(u_j)$, with
$f(t)= \frac1{1-t}$.  If only terms of the form $u_i-u_j$ or $u_i-1$
occur, then we can test the emptiness of~$R$ by examining the possible
orders on $(1, u_1, u_2, \ldots, u_n)$. We propose to handle the terms
of the form $u_i-f(u_j)$ in the same way, by exploring the orders that
can arise on $(1,u_1,f(u_1), u_2, f(u_2), \ldots, u_n,f(u_n))$.

\medskip

The main difficulty in this approach is to restrict the exploration to
the orders that can be realized by a sequence of the form
$(1,u_1,f(u_1), u_2, f(u_2), \ldots, u_n,f(u_n))$. This turns out to
be easier if we extend the lifting to
\[ \Lambda: \left\{\begin{array}{rcl} (\R\setminus\{0,1\})^{n} & \to & \R^{3n}
\\ (u_1,u_2, \ldots, u_n) & \mapsto & (u_1,f(u_1),f^{(2)}(u_1),
\ldots, u_n, f(u_n),f^{(2)}(u_n)).\end{array}\right.\]
This extended lifting allows to take advantage of the facts that
$f^{(3)}=f \circ f \circ f$ is the identity, that $f$ permutes
circularly the intervals $(-\infty,0)$, $(0,1)$ and $(1, \infty)$, and
that $f$ is increasing on each of them. In
Proposition~\ref{p:cellR3n}, we essentially show that \emph{an order on the $3n$
  lifted variables can be realized by a point of~$\Lambda(\R^{3n})$ if
  and only if it is compatible with the action of~$f$, as captured by
  these properties.}

\subparagraph{Algorithm.}

Our second contribution is an algorithm that puts the combinatorial
lifting in practice, and decides if a triple of permutations is
realizable in $O\pth{6^nn^{10}}$ time and $O(n^2)$ space in the
worst-case. We provide an implementation in Python (see  Appendices~\ref{a:raw} and~\ref{a:code}).

\subparagraph{New geometric results.}

Our remaining contributions are new geometric results obtained with
the aid of our implementation. A first systematic exploration reveals:

\begin{theorem}\label{thm:five}
  Every triple of permutations of size $5$ is geometrically realizable in~$\R^3$.
\end{theorem}

\noindent
The smallest known triple of geometric permutations forbidden in
$\R^3$ has size~$6$ (see Section~\ref{s:discussion}), so
Theorem~\ref{thm:five} proves that it is minimal. We also
  obtained the complete list of forbidden triples of size~$6$ (see
  Appendix~\ref{a:forbid6}). Interestingly, although everything is realizable
up to size~$5$, something can be said on geometric permutations of
size~$4$. Recall that the \emph{side operator} $(pq) \odot (rs)$ of
the two lines $(pq)$ and $(rs)$, oriented respectively from $p$ to $q$
and from $r$ to $s$, is the orientation of the tetrahedron $pqrs$; it
captures the mutual disposition of the two lines. We prove:

\begin{theorem}\label{thm:so}
  Let $\ell_1$ and $\ell_2$ be two oriented lines intersecting four
  pairwise disjoint convex sets in the order~$1234$. Any oriented line
  $\ell_3$ that intersects those four sets in the order $2143$
  satisfies $\ell_1 \odot \ell_3 = \ell_2 \odot \ell_3$.
\end{theorem}

\noindent
The pattern $(1234,2143)$ is known to be forbidden in some cases (see
Section~\ref{s:discussion}), but this is the first condition valid for
arbitrary disjoint convex sets. We could prove Theorem~\ref{thm:so}
because our algorithm solves a more constrained problem than just
realizability of permutations. Given three lines in general position
in $\R^3$, there is a unique parallelotope with three disjoint edges
supported on these three lines (see
Figure~\ref{fig:boxedlines}). Combinatorial lifting, and therefore our
algorithm, can decide whether three permutations can be realized
\emph{with the vertices of that parallelotope in prescribed positions
  in the permutations}.

\begin{figure}[t]
  \begin{center}\includegraphics[page=1]{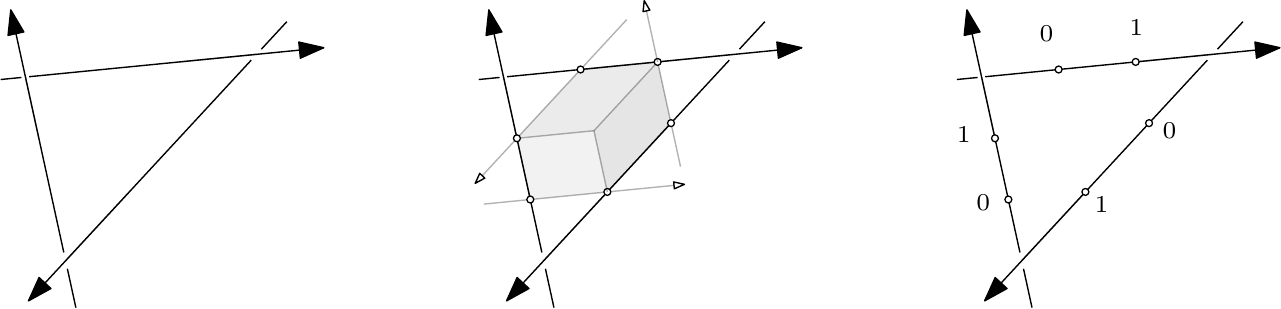}\end{center}
  \caption{Three skew lines (left), the parallelotope (middle) and the marked points (right).\label{fig:boxedlines}}
\end{figure}

We label the vertices of the parallelotope with $\0$ and $\1$ as in
Figure~\ref{fig:boxedlines} and work with permutations where two extra
elements, $\0$ and $\1$, are inserted; we call them \emph{tagged
  permutations}. We examine triples of tagged permutations realizable
on a canonical system of lines (see Equation~\eqref{eq:canonical}),
and characterize those minimally unrealizable up to size~$4$
(Proposition~\ref{p:tag}); for size~$2$ and~$3$, we
provide independent, direct, geometric proofs of unrealizability
(Section~\ref{s:analysis}). We conjecture that no other minimally
unrealizable triples of tagged permutations exist, and verified this
experimentally up to size~$6$ (not counting $\0$ and $\1$). A weaker conjecture is:

\begin{conjecture}\label{c:poly}
  There exists a polynomial time algorithm that decides the geometric
  realizability of a triple of permutations of size $n$ in $\R^3$.
\end{conjecture}

\subsection{Discussion and related work}
\label{s:discussion}

We now put our contribution in context, starting with motivations for
studying geometric permutations.

\subparagraph{Geometric transversals.}

In the 1950's, Gr\"unbaum~\cite{g-ct-58} conjectured that, given a
family of disjoint translates of a convex figure in the plane, if
every five members of the family can be met by a line, then there
exists a line that meets the entire family. (Such a statement, if
true, is an example of \emph{Helly-type theorem}.) Progress on
Gr{\"u}nbaum's conjecture was slow until the 1980's, when the notion
of geometric permutations of families of convex sets was
introduced~\cite{KLZ85,katchalski1986geometric}. Their systematic study
in the plane was refined by Katchalski~\cite{katchalski-grunbaum} in
order to prove a weak version of Gr\"unbaum's conjecture (with $128$
in place of $5$). Tverberg~\cite{t-pgcct-89} soon followed up with a
proof of the conjecture, again using a careful analysis of planar
geometric permutations. This initial success and further conjectures
about Helly-type theorems stimulated a more systematic study of
geometric permutations realizable under various geometric
restrictions; \emph{cf.}~\cite{holmsenwenger} and the references
therein.

Another motivation to study geometric permutations comes from
computational geometry, more precisely the study of geometric
structures such as \emph{arrangements}. There, geometric permutations
appear as a coarse measure of complexity of the space of line
transversals to families of sets, and relates to various algorithmic
problems such as ray-shooting or smallest enclosing cylinder
computation~\cite[$\mathsection 7.6$]{pach2009combinatorial}. From
this point of view, the main question is to estimate the maximum
number $g_d(n)$ of distinct geometric permutations of $n$ pairwise
disjoint convex sets in~$\R^d$. Broadly speaking, while $g_2(n)$ is
known to equal $2n-2$~\cite{es-mnwsn-90}, even the order of magnitude
of $g_d(n)$ as $n \to \infty$ is open for every $d \ge 3$; the gap is
between $\Omega(n^{d-1})$ and $O(n^{2d-3}\log n)$. Bridging this gap
has been identified as an important problem in discrete
geometry~\cite[$\mathsection 7.6$]{pach2009combinatorial}, yet, over
the last fifteen years, the only progress has been an improvement of
the upper bound from $O(n^{2d-2})$ down to $O(n^{2d-3}\log n)$;
moreover, while the former bound follows from a fairly direct
argument, the latter is a technical tour de force
\cite{rubin2012improved}. We hope that a better understanding of small
forbidden configurations will suggest new approaches to this question.

\subparagraph{Geometric realizability problems.}

Combinatorial structures that arise from geometric configurations such
as arrangements, polytopes, or intersection graphs are classical
objects of enquiry in discrete and computational geometry (see
\emph{e.g.}~\cite[$\mathsection$~1, 5, 6, 10, 15, 17,
  28]{handbook}). We are concerned here with the \emph{membership
  testing problem}: given an instance of a combinatorial structure,
decide if there exists a geometric configuration that induces it. Such
problems can be difficult: for instance, deciding whether a given
graph can be obtained as intersection graph of segments in the plane
is NP-hard~\cite{kratochvil1994intersection}.

A natural approach to membership testing is to parameterize the
candidate geometric configuration and express the combinatorial
structure as conditions on these parameters. This often results in a
semi-algebraic set. In the real-RAM model\footnote{See the same
  reference for a similar bound in the bit model.}, the emptiness of a
semi-algebraic set in $\R^d$ with real coefficients can be tested in
time $(nD)^{O(d)}$~\cite[Prop. 4.1]{renegar1992computational}, where
$n$ is the number of polynomials and $D$ their maximum degree. (Other
approaches exist but have worse complexity bounds,
see~\cite{davenport1988real,lombardi2014elementary,kaltofen2012exact,cucker2017computing}). Given
three permutations of size~$n$, we describe their realizations as a
semi-algebraic set defined by $O(n^2)$ cubic polynomials in $n$
variables; the above method thus has complexity $n^{O(n)}$, making our
$O\pth{6^nn^{10}}$ solution competitive in theory. Practical
effectiveness is usually difficult to predict as it depends on the
geometry of the underlying algebraic surfaces; for example, deciding
if two geometric permutations of \emph{size four} are realizable by
disjoint unit balls in $\R^3$ was recently checked to be out of
reach~\cite{ha2016geometric}.

Some geometric realizability problems, for example the recognition of
unit disk graphs, are $\exists\R$-hard~\cite{schaefer2017fixed} and
therefore as difficult from a complexity point of view as deciding the
emptiness of a general semi-algebraic set. We do not know whether
deciding the emptiness of semi-algebraic sets amenable to our
combinatorial lifting remains $\exists\R$-hard; we believe, however,
that deciding if a triple of permutations is realizable is not
(\emph{c.f.} Conjecture~\ref{c:poly}).

\subparagraph{Forbidden patterns.}

A dimension count shows that any $k$ permutations are realizable
in~$\R^{2k-1}$. Our most direct predecessor is the work of Asinowski
and Katchalski~\cite{Asinowski2005} who proved that this argument is
sharp by constructing, for every $k$, a set of $k$ permutations that
are not realizable in $\R^{2k-2}$. They also showed that the triple
$(123456, 321654, 246135)$ is not realizable in $\R^3$, a fact that
follows easily from our list of obstructions.

In a sense, our work tries to generalize some arguments previously
used to analyze geometric permutations in the plane. For example, the
(standard) proof that the pair $(1234, 2143)$ is non-realizable as
geometric permutations in $\mathbb{R}^2$ essentially analyzes tagged
permutations. Indeed, if we augment the permutations by an additional
label $\0$ marking the intersection of the lines realizing the two
orders, we get that $(\0ab, \0ba)$, $(\0ab, ab\0)$, $(ba\0, \0ba)$ and $(ba\0, ab\0)$ are forbidden, and there is nowhere
to place $\0$ in $(1234, 2143)$. (See also
Section~\ref{s:interpretation} and Observation \ref{planar-obs}.) We
should, however, emphasize that already in $\R^3$ the geometry is much
more subtle.

Forbidden patterns were used to bound the number of geometric
permutations for certain restricted families of convex sets. For
pairwise disjoint translates of a convex planar
figure~\cite{katchalski-grunbaum, t-pgcct-89, tverberg91}, it is known
that a given family can have at most three geometric permutations, and
the possible sets of realizable geometric permutations have been
characterized. The situation is similar for families of pairwise
disjoint unit balls in $\R^d$. Here, an analysis of forbidden patterns
in geometric permutations showed that a given family can have at most
a constant number of geometric permutations (in fact only two if the
family is sufficiently large) \cite{sms00, CGN05,
  ha2016geometric}. Another example is \cite{shakhar-origins}, where
it is shown that the maximum number of geometric permutations for
convex objects in $\mathbb{R}^d$ induced by lines that pass through
the origin, is in $\Theta(n^{d-1)}$. The restriction that the lines
pass through the origin, allows them to deal with permutations
augmented by one additional label, and their argument relies on the
forbidden tagged pattern $(\0ab,\0ba)$ \cite[Lemma
  2.1]{shakhar-origins}.

In these examples, the bounds use highly structured sets of forbidden
patterns. In general, one cannot expect polynomial bounds on the sole
basis of excluding a handful of patterns; for instance it is not hard
to construct an exponential size family of permutations of $[n]$ which
avoids the pattern $(1234, 2143)$. Such questions are well-studied in
the area of ``pattern-avoidance'' and usually the best one could hope
for is an exponential upper bound on the size of the family
\cite{MarcusTardos}.

\section{Semi-algebraic parameterization}
\label{s:parameter}

Let $P = (\pi_1, \pi_2,\pi_3)$ denote a triple of permutations of
$\{1,2,\ldots, n\}$. We now describe a semi-algebraic set that is nonempty
if and only if $P$ has a geometric realization in~$\R^3$.

\subparagraph{Canonical realizations.}

We say that a geometric realization of $P$ is \emph{canonical} if the oriented line transversals are
\begin{equation}\label{eq:canonical}
  \ell_x = \pth{\begin{matrix}0\\1\\0\end{matrix}} + \R \pth{\begin{matrix}1\\0\\0\end{matrix}}, \quad
\ell_y = \pth{\begin{matrix}0\\0\\1\end{matrix}} + \R \pth{\begin{matrix}0\\1\\0\end{matrix}}, \hbox{ and }\quad
\ell_z = \pth{\begin{matrix}1\\0\\0\end{matrix}} + \R \pth{\begin{matrix}0\\0\\1\end{matrix}},
\end{equation}
and if the convex sets are triangles with vertices on $\ell_x$, $\ell_y$ and $\ell_z$. 

\begin{lemma}\label{l:canonical}
If $P$ is geometrically realizable in $\R^3$, then it has a canonical realization, possibly after reversing some of the permutations $\pi_i$.
\end{lemma}
\begin{proof}
  Consider a realization of $P$ by three lines and $n$ pairwise
  disjoint sets. For each convex set we select a point from the
  intersection with each of the lines and replace it by the (possibly
  degenerate) triangle spanned by these points. This realizes $P$ by
  {\em compact} convex sets. By taking the Minkowski sum of each set
  with a sufficiently small ball, the sets remain disjoint and the
  line intersect the sets in their interior. We may now perturb the
  lines into three lines $L_x$, $L_y$ and $L_z$ that are pairwise skew
  and not all parallel to a common plane. We then, again, crop each
  set to a triangle with vertices on $L_x$, $L_y$ and $L_z$.

  We now use an affine map to send our three lines to $\ell_x$,
  $\ell_y$ and $\ell_z$. An affine transform is defined by $12$
  parameters and fixing the image of one line amounts to four linear
  conditions on these parameters; these constraints determine a unique
  transform because the lines are in general position. Note, however,
  that the oriented line $L_x$ is mapped to either $\ell_x$ or
  $-\ell_x$, so $\pi_1$ may have to be reversed; the same applies to
  the permutations $\pi_2$ and $\pi_3$.
\end{proof}

\noindent
We equip the line $\ell_x$ (resp. $\ell_y$, $\ell_z$) with the
coordinate system obtained by projecting the $x$-coordinate
(resp. $y$-coordinate, $z$-coordinate) of $\R^3$. This parameterizes the space of canonical realizations by $\R^{3n}$. Specifically, we equip $\R^{3n}$ with a coordinate system $(O,x_1,x_2, \ldots, x_n, y_1, y_2, \ldots,$ $y_n, z_1, z_2, \ldots, z_n)$ and for any point $c \in \R^{3n}$ we put
\[ \Tr(c) = \{\mbox{conv}\{X_i,Y_i,Z_i\}\}_{1 \le i \le n}, \quad \hbox{where} \quad
X_i=\left(\begin{matrix}x_i\\1\\0\end{matrix}\right), \quad Y_i =
  \left(\begin{matrix}0\\y_i\\1\end{matrix}\right), \hbox{ and } \quad
    Z_i =\left(\begin{matrix}1\\0\\z_i\end{matrix}\right).
\]
Each element of $\Tr(c)$ is thus a triangle with a vertex on each of
$\ell_x$, $\ell_y$ and $\ell_z$. We define:
\[R = \{c \in \R^{3n} \colon \Tr(c) \hbox{ consists of disjoint
  triangles and realizes } P\}.\]
The triple $P$ is realizable if and only if $R$ is non-empty.

\subparagraph{Triangle disjointedness.}

We now review an algorithm of Guigue and
Devillers~\cite{guigue2003fast} to decide if two triangles are
disjoint, and use it to formulate the condition that two triangles
$X_iY_iZ_i$ and $X_jY_jZ_j$ be disjoint as a semi-algebraic condition
on $x_i, \ldots, z_j$.

\medskip

The algorithm and our description are expressed in terms of
orientations, where the \emph{orientation} of four points $p,q,r,s \in
\R^3$ is
\[ [p,q,r,s] \eqdef \sign \det \pth{
\begin{array}{cccc}
 x_p & x_q & x_r & x_s\\
 y_p & y_q & y_r & y_s\\
 z_p & z_q & z_r & z_s\\ 
 1 & 1& 1& 1 
\end{array}}.\]
Intuitively, the orientation indicates whether point $s$ is ``above''
(+1), on (0), or ``below'' (-1) the plane spanned by $p,q,r$, where
above and below refer to the orientation of the plane that makes the
directed triangle $pqr$ positively oriented. We only consider
orientations of non-coplanar quadruples of points, so orientations
take values in $\{\pm1\}$.

\medskip

If one triangle is on one side of the plane spanned by the other, then
the triangles are disjoint. We check this by computing
\[ v(i,j) \eqdef  \left(
\begin{matrix}
  [X_i,Y_i,Z_i,X_j]\\
  [X_i,Y_i,Z_i,Y_j]\\
  [X_i,Y_i,Z_i,Z_j]\\
  [X_j,Y_j,Z_j,X_i]\\
  [X_j,Y_j,Z_j,Y_i]\\
  [X_j,Y_j,Z_j,Z_i]
\end{matrix}
\right) \in \{-1,1\}^6\]
and testing if $v(i,j)_1=v(i,j)_2=v(i,j)_3$ or
$v(i,j)_4=v(i,j)_5=v(i,j)_6$. If this fails, then we rename
$\{X_i,Y_i,Z_i\}$ into $\{A_i,B_i,C_i\}$ and $\{X_j,Y_j,Z_j\}$ into
$\{A_j,B_j,C_j\}$ so~that
\[ \pth{
  \begin{matrix}
    [A_i,B_i,C_i,A_j]\\
    [A_i,B_i,C_i,B_j]\\
    [A_i,B_i,C_i,C_j]\\
    [A_j,B_j,C_j,A_i]\\
    [A_j,B_j,C_j,B_i]\\
    [A_j,B_j,C_j,C_i]
\end{matrix}  } = \pth{\begin{matrix} 1 \\ -1 \\ -1\\ 1 \\ -1 \\ -1\end{matrix} }.\]
Then, the triangles are disjoint if and only if $[A_i,B_i,A_j,B_j] =
1$ or $[A_i,C_i,C_j,A_j] = 1$~\cite{guigue2003fast}. The renaming is
done as follows.  Since the first test is inconclusive, the plane
spanned by a triple of points separates the other triple of points. We
let $(A_i,B_i,C_i)$ be the circular permutation of $(X_i,Y_i,Z_i)$
such that $A_i$ is separated from $B_i$ and $C_i$ by the plane spanned
by $X_j$, $Y_j$, and $Z_j$. We let $(A_j,B_j,C_j)$ be the circular
permutation of $(X_j,Y_j,Z_j)$ such that $A_j$ is separated from $B_j$
and $C_j$ by the plane spanned by $A_i$, $B_i$, and $C_i$. If
$[A_i,B_i,C_i,A_j] = -1$ then we exchange $B_i$ and $C_i$. If
$[A_j,B_j,C_j,A_i] = -1$ then we exchange $B_j$ and $C_j$.

\subparagraph{Semi-algebraicity.}

Every step in the Guigue-Devillers algorithm can be expressed as a
logical proposition in terms of orientation predicates which are, when
specialized to our parameterization, conditions on the sign of
polynomials in the coordinates of $c$. Checking that each of $\ell_x$,
$\ell_y$ and $\ell_z$ intersects the triangles in the prescribed order
amounts to comparing coordinates of $c$. Altogether, the set $R$ is a
semi-algebraic subset of $\R^{3n}$.

\section{Combinatorial lifting}\label{s:combilift}

We now explain how to test combinatorially the emptiness of our
semi-algebraic set $R$.

\subparagraph{Definitions.}

We start by decomposing each orientation predicate used in the
definition of $R$ as indicated in Table~\ref{t:decomp}. For the last three rows, this is not a
factorization since one of the factors is 
of the form $u - f(v)$ where $f:t \mapsto
\frac1{1-t}$.

\begin{table}[t]
  \begin{center}
  \begin{tabular}[padding=25px]{|c|c|c|}
    \hline
    Orientation & Determinant & Decomposition \\
    \hline
    \hline
    $[X_a,X_b,Y_c,Y_d]$ & $(x_a-x_b)(y_c-y_d)$ & $(x_a-x_b)(y_c-y_d)$\\
    $[X_a,X_b,Z_c,Z_d]$ & $(x_a-x_b)(z_c-z_d)$ & $(x_a-x_b)(z_c-z_d)$  \\
    $[Y_a,Y_b,Z_c,Z_d]$ & $(y_a-y_b)(z_c-z_d)$ & $(y_a-y_b)(z_c-z_d)$ \\
    \hline
    $[X_a,X_b,Y_c,Z_d]$ & $(x_a-x_b)(y_cz_d-z_d+1)$ & $(x_a-x_b)(y_c-1)\pth{z_d-\frac1{1-y_c}}$\\
    $[X_a,Y_b,Y_c,Z_d]$ & $(y_b-y_c)(x_a-x_az_d-1)$ & $- (y_b-y_c)(z_d-1)\pth{x_a-\frac1{1-z_d}}$\\
    $[X_a,Y_b,Z_c,Z_d]$ & $(z_c-z_d)(x_ay_b+1-y_b) $& $(z_c-z_d)(x_a-1)\pth{y_b-\frac1{1-x_a}}$\\
    \hline
  \end{tabular}
  \end{center}
  \caption{Orientation predicates used in the Guigue-Devillers
    algorithm when specialized to points from $\ell_x$, $\ell_y$ and
    $\ell_z$.\label{t:decomp}}
\end{table}

In light of the third column of Table~\ref{t:decomp}, it may seem
natural to ``linearize'' the problem by considering the map 
$(x_1,x_2, \ldots, z_n) \mapsto (x_1,f(x_1), x_2, f(x_2), \ldots, z_n, f(z_n))$
 from $\R^{3n}$ to $\R^{6n}$. Indeed, the order on the lifted coordinates and~$1$ determines the
sign of all polynomials defining~$R$. We must, however, identify the
orders on the coordinates in $\R^{6n}$ that can be realized by lifts
of points from $\R^{3n}$. Perhaps surprisingly, the task gets easier
if we lift to even higher dimension. For convenience we let $\R_*
\eqdef \R\setminus \{0,1\}$. The lifting map we use is:
\[ \Lambda: \left\{\begin{array}{rcl} \R_*^{3n} & \to & \R^{9n} \\
(x_1,x_2, \ldots, z_n) & \mapsto & \left( x_1,f(x_1),f^{(2)}(x_1), x_2,
\ldots, z_n, f(z_n),f^{(2)}(z_n)\right)\end{array}\right.\]
To determine the image of $\Lambda\pth{\R_*^{3n}}$, we will use the
following properties of $f$:

\begin{claim}\label{c:f}
  $f^{(3)} = f \circ f \circ f$ is the identity on $\R_*$, $f$
  permutes the intervals $(-\infty,0)$, $(0,1)$ and $(1,+\infty)$
  circularly, and $f$ is monotone on each of these intervals.
\end{claim}

\noindent
Let us denote the points of $\mathbb{R}^{9n}$ by vectors $(\ttt_1,
\ttt_2, \ldots, \ttt_{9n})$. We next ``lift'' the semi-algebraic
description of $R$:

\begin{enumerate}
\item We pick a Boolean formula $\phi$ describing $R$ in terms of
  orientations (for the triangle disjointedness) and comparisons of
  coordinates (for the geometric permutations).

\item We decompose every orientation predicate ocurring in $\phi$ as
  in the third row of Table~\ref{t:decomp}.

\item We then construct another Boolean formula $\psi$ by substituting\footnote{For example, with $n=3$, the product 
$ (x_1 - x_2) (y_2 - 1) \left(z_3 - \tfrac{1}{1-y_2}\right) = (x_1-x_2) (y_2 - 1) (z_3 - f(y_2))$
appearing in $\phi$ is translated in $\psi$ as $(\ttt_1 - \ttt_4)
(\ttt_{13}-1) (\ttt_{25} - \ttt_{14})$.}
  in $\phi$ every $f(x_1)$ by the variable $\ttt_2$ (to which it is
  mapped under $\Lambda$). We similarly substitute every $f(x_i)$,
  $f(y_i)$ and $f(z_i)$, then every remaining $x_i$, $y_i$ and $z_i$
  by the corresponding variable $\ttt_*$.

\item We let $S \subset \R^{9n}$ be the (semi-algebraic) set of points
  that satisfy $\psi$.
\end{enumerate}

\noindent
We finally let $\h$ denote the arrangement in $\R^{9n}$ of the
set of hyperplanes:
\[\left\{ \ttt_i = \ttt_j \right\}_{1 \le i<j\le 9n} \cup \{\ttt_i = 0\}_{1 \le i
  \le 9n} \cup \{\ttt_i = 1\}_{1 \le i \le 9n}.\]
Note that the full-dimensional (open) cells in $\h$ are in bijection
with the total orders on $\{0,1,\ttt_1, \ldots, \ttt_{9n}\}$ in which
$0$ comes before $1$. We write $\prec_A$ for the order associated with
a full-dimensional cell $A$ of $\h$.
  
\begin{lemma}\label{l:lift1}
  Every full-dimensional cell of $\h$ is disjoint from or contained in
  $S$. Moreover, $R$ is nonempty if and only if there exists a
  full-dimensional cell of $\h$ that is contained in $S$ and
  intersects $\Lambda\pth{\R_*^{3n}}$.
\end{lemma}
\begin{proof}
  The set $S$ is defined by the positivity or negativity of
  polynomials, each of which is a product of terms of the form
  $(\ttt_i - \ttt_j)$ or $(\ttt_i - 1)$. The first statement thus
  follows from the fact the coordinates of all points in a
  full-dimensional cell realize the same order on $\{0,1,\ttt_1,
  \ldots, \ttt_{9n}\}$.  By the perturbation argument used in the
  proof of Lemma~\ref{l:canonical}, if $R$ is non-empty, then it
  contains a point with no coordinate in $\{0,1\}$. Thus, $R$ is non-empty if and only if $\Lambda(R)$ is
  non-empty. The construction of $S$ ensures that $\Lambda(R) = S \cap
  \Lambda\pth{\R_*^{3n}}$. Again, a perturbation argument ensures that
  if $\Lambda(R)$ is nonempty, it contains a point outside of the
  union of the hyperplanes of $\h$. The second statement follows.
\end{proof}

\subparagraph{Zone characterization.}

Inspired by Lemma~\ref{l:lift1}, we now characterize the orders
$\prec_A$ such that $A$ intersects $\Lambda\pth{\R_*^{3n}}$. We
split the $9n$ variables $\ttt_1, \ttt_2, \ldots, \ttt_{9n}$ into $3n$
blocks of three consecutive variables $\ttt_{3i+1}, \ttt_{3i+2}$,
$\ttt_{3i+3}$ (representing $x_i, f(x_i), f^{(2)}(x_i)$ for $0\leq
i<n$, $y_i, f(y_i), f^{(2)}(y_i)$ for $n\leq i <2n$, and $z_i, f(z_i),
f^{(2)}(z_i)$ for $2n\leq i <3n$). We also define an operator $\ff$
that shifts the variables cyclically within each individual block:
\[ \ff(\ttt_{3i+1}) =  \ttt_{3i+2}, \quad  \ff(\ttt_{3i+2}) =  \ttt_{3i+3} \quad \hbox{and} \quad   \ff(\ttt_{3i+3})  = \ttt_{3i+1}. \]

\noindent
By convention, $\ff^0$ means the identity. The fact that $\ff$
mimicks, symbolically, the action of~$f$ yields the following
characterization.

\begin{proposition}\label{p:cellR3n}
  A full-dimensional cell $A$ of $\h$ intersects
  $\Lambda\pth{\R_*^{3n}}$ if and only if

\begin{enumerate}[(i)]
\item For any $0\leq i <3n$, there exists $j \in \{0,1,2\}$ s. t. \[\ff^{(j)}(\ttt_{3i+1}) \prec_A 0 \prec_A 
  \ff^{(j+1)}(\ttt_{3i+1}) \prec_A 1 \prec_A \ff^{(j+2)}(\ttt_{3i+1}).\]
\item For any $1\leq i,j \le 9n$, \qquad $\left\{\ttt_i \prec_A \ttt_j \quad \text{and} \quad \ff(\ttt_j) \prec_A
    \ff(\ttt_i)\right\} \quad \Rightarrow \quad \ttt_i \prec_A 1 \prec_A \ttt_j.$
\end{enumerate}
\end{proposition}
\begin{proof}
  Let us first see why the conditions are necessary. Let $c = (x_1,
  x_2, \ldots, z_n) \in \R_*^{3n}$ such that $\Lambda(c) \in A$. Fix
  some $0\leq i < n$. As $j$ ranges over $\{0,1,2\}$, the coordinate
  $\ff^{(j)}(\ttt_{3i+i})$ of $\Lambda(c)$ ranges over $\{x_i,f(x_i),
  f^{(2)}(x_i)\}$, and Condition~(i) holds because $f$ permutes the
  intervals $(-\infty,0)$, $(0,1)$ and $(1,+\infty)$ circularly. The
  cases $n \le i < 3n$ are similar. Condition~(ii) follows in a
  similar manner from the fact that $f$ permutes the intervals
  $(-\infty,0)$, $(0,1)$ and $(1,+\infty)$ circularly and is
  increasing on each of them.

  \bigskip
  
  To examine sufficiency we need some notations. We let $\var =
  \{0,1,\ttt_1, \ldots, \ttt_{9n}\}$. Given an order $\prec$ on $\var$
  and two elements $a,b \in \var$ we write $(a,b)_{\prec} \eqdef \{ c
  \in \var: a \prec c \prec b\}$. We also write $(\cdot,a)_{\prec}$
  for the set of elements smaller than $a$, and $(a, \cdot)_{\prec}$
  for the set of elements larger than~$a$, and $[a,b)_{\prec}$,
    $[a,b)_{\prec}$ or $[a,b]_{\prec}$ to include one or both bounds
      in the interval.

  \medskip

  Let $\prec_*$ be an order on $\{0,1,\ttt_1, \ldots, \ttt_{9n}\}$
  such that $0 \prec_* 1$. By Condition~(i), $(1,\cdot)_{\prec_*}$ has
  size $3n$, so let us write $(1,\cdot)_{\prec_*} = \{b_1, b_2, \ldots,
  b_{3n}\}$ with $1 {\prec_*} b_1 {\prec_*} b_2 {\prec_*} \ldots
  {\prec_*} b_{3n}$. Condition~(i) also ensures that for every $0\leq
  i< 3n$, exactly one of $\{\ttt_{3i+1}, \ff(\ttt_{3i+1}),
  \ff^{(2)}(\ttt_{3i+1})\}$ belongs to $(1,\cdot)_{\prec_*}$.  Hence,
  for every $0\leq i<3n$ there are uniquely defined integers $0\leq
  \alpha(i)\leq 2$ and $1\leq \beta(i)\leq 3n$ such that $b_{\beta(i)}
  = \ff^{\alpha(i)}(\ttt_{3i+1})$.

  \medskip

  We next pick $3n$ real numbers $1 < r_1 < r_2 < \ldots < r_{3n}$, put

  \[\begin{array}{llll}
     x_i &= &f^{3-\alpha(i)}\pth{r_{\beta(i)}},& \text{ for } 0\leq i<n  \\
     y_i &= &f^{3-\alpha(i)}\pth{r_{\beta(i)}},& \text{ for } n\leq i<2n  \\
     z_i &= &f^{3-\alpha(i)}\pth{r_{\beta(i)}},& \text{ for } 2n\leq i<3n,
  \end{array}\]

  \noindent
  and let $p = (x_1, \ldots, x_n, y_1, \ldots, y_n, z_1, \ldots, z_n)
  \in \mathbb{R}^{3n}$. Note that $\Lambda(p)$ lies in a
  full-dimensional cell of the arrangement $\h$; let us denote it by
  $A$.

  \medskip

  Now, $0$ precedes $1$ in both $\prec_*$ and $\prec_A$. Also, $[1,\cdot)_{\prec_*} =
    [1,\cdot)_{\prec_A}$ and the two orders coincide on that interval by construction of~$p$.
  Remark that $\ff$ acts similarly for both orders:

\begin{itemize}
\item $\ff$ maps $[1,\cdot)_{\prec_*}$ to $(\cdot,0]_{\prec_*}$
  increasingly for $\prec_*$ by Conditions~(i) and~(ii).

\item $\ff$ maps $[1,\cdot)_{\prec_A}$ to $(\cdot,0]_{\prec_A}$
  increasingly for $\prec_A$ by definition of $p$ and Claim~\ref{c:f}.
\end{itemize}
    
\noindent
We therefore also have $(\cdot,0]_{\prec_*} = (\cdot,0]_{\prec_A}$ and
the orders coincide on that interval as well. The same argument
applied to $\ff^2$ shows that $[0,1]_{\prec_*} = [0,1]_{\prec_A}$
and that the two orders coincide on that interval as
well. Altogether, $\prec_*$ and $\prec_A$ coincide.
\end{proof}

\section{Geometric interpretation of the combinatorial lifting}
\label{s:interpretation}

Let us take a moment to consider the geometric meaning of our parameterization and lifting.

\subparagraph{Summary.} 

In short, Section~\ref{s:parameter} reduced our initial problem to the
more specialized one of realizing a triple of permutations by the
lines $\ell_x$, $\ell_y$, $\ell_z$ and triangles with a vertex on each
line (Lemma~\ref{l:canonical}). The canonical coordinate system of
$\R^3$ induces natural coordinate systems on each of these lines, and
we use it to parameterize the positions of the triangles'
vertices. The set of parameters that correspond to a geometric
realization of our three permutations is then seen to form a
semi-algebraic set~$R$. In Section~\ref{s:combilift}, we introduced
the lift $\Lambda$ to analyze $R$ combinatorially provided that we
specify some extra information: the comparisons between each variable
and the constants $0$ and $1$ (Lemma~\ref{l:lift1} and
Proposition~\ref{p:cellR3n}).

\subparagraph{Tags.}

Let us reformulate this extra information geometrically. Consider
three oriented lines in~$\R^3$. Each line can be translated so as to
simultaneously intersect the other two lines; we mark these
intersection points on the two (non-translated) lines. Altogether, we
collect two points per line, which we label $\0$ and $\1$ with the
convention that $\0$ comes before $\1$ in the orientation of the
line. Equivalently, these six points can be obtained by considering
the unique parallelotope that has three disjoint edges supported by
the lines, and marking the vertices of these edges (see
Figure~\ref{fig:2d}-left). Now, specifying how $\ttt_{3i+1}$ (say for
$0 \le i < n$) compares to $0$ and $1$ in $\var$ is equivalent to
specifying where $x_i$ lies compared to $\0$ and $\1$ on~$\ell_x$. The
combinatorial lifting therefore highlights that the position of the
triangle vertices' with respect to the parallelotope is a useful
information for checking geometric realizability.

\begin{figure}[t]
  \begin{center}\includegraphics[page=2]{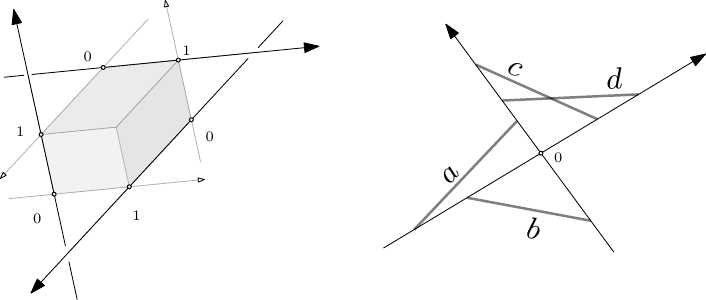}\end{center}
  \caption{Left: The marked point associated with three skew lines in
    $\R^3$. Right: In $\R^2$, if two lines crossing at $\0$ intersect
    segments $\{a,b,c,d\}$ in the (extended) orders $(ab{\0}cd,
    b{\0}adc)$, then the sets $c$ and $d$ must cross.\label{fig:2d}}
\end{figure}

\subparagraph{Analogy with the planar case.}

A similar observation was used in the plane (see
\emph{e.g.}~\cite{katchalski1985geometric}). Consider two lines in
$\R^2$, crossing in $\0$, and a family of segments, where each segment
has an endpoint on each line (see Figure~\ref{fig:2d}-right). Every
segment lies in a (closed) quadrant formed by the lines, and two
segments intersect if and only if they lie in the same quadrant and
appear in different orders when seen ``from $\0$''. The quadrant
containing a given segment is determined by the positions of that
segment's endpoints with respect to $\0$. As mentioned in the
introduction, a simple case analysis then yields that $(1234,2143)$
has no geometric realization in the plane. (See~\cite[Figure
  3.4]{asinowski1999geometric} for another example of such case
analysis.)

\subparagraph{Changes between two and three
dimensions.}

In the plane, given two permutations and the position of the crossing
point of the two lines, either all choices of positions yield pairwise
disjoint segments, or none of them does. In~$\R^3$, the polynomials
describing whether two triangles are disjoint requires in some cases
to compare pairs $\{x_i, f(z_j)\}$, $\{y_i, f(x_j)\}$ or
$\{z_i,f(y_j)\}$ (Table~\ref{t:decomp}). This makes it possible for
two pairs of triangles, one crossing and the other disjoint, to
realize the same three tagged permutations, \emph{i.e.}  have their
vertices in the same position relative to the parallelotope of the
lines.

\subparagraph{Tagged permutations and patterns.}

Formally, we define a \emph{tagged permutation} as a permutation of
$\{\0, \1, 1,2, \ldots, n\}$ in which $\0$ precedes $\1$. We call a
triple of tagged permutations a \emph{tagged pattern}. A
\emph{canonical realization} of a tagged pattern is a set of
triangles, with vertices on $\ell_x$, $\ell_y$ and $\ell_z$, such that
$\ell_x$ (resp. $\ell_y$, $\ell_z$) intersects the triangles in the
first (resp. second, third) permutation and such that the tagged
corners of the parallelotope appear in the right position on each
line. 

Our experiments will use two more notions. Two tagged patterns are
\emph{equivalent} for canonical realizability if one can be
transformed into the other by (i) relabeling the symbols other than
$\0$ and $\1$ bijectively, and (ii) applying a circular permutation to
the triple. A tagged pattern is \emph{minimally forbidden} if it has
no canonical realization, and deleting any symbol other than $\0$
and~$\1$ from the three tagged permutations produces a tagged pattern
which has a canonical realization.

\section{Algorithm}
\label{s:algorithm}

We now present an algorithm that takes a tagged pattern as input and
decides if it admits a canonical realization. Our initial problem of
testing the geometric realizability of a triple of permutations of size $n$
reduces to $8 \binom{n+2}{2}^3$ instances of that problem.

\subsection{Outline}

Following Sections~\ref{s:parameter} and~\ref{s:combilift}, we search
for an order on $\{0,1,\ttt_1, \ttt_2, \ldots, \ttt_{9n}\}$ satisfying
the conditions of Proposition~\ref{p:cellR3n} and the formula $\psi$
(which defines $S$). To save breath, we call such an order
\emph{good}.  We say that \emph{triangles $i$ and $j$ are disjoint in
  a partial order} $\po$ if for every $c \in \R^{3n}$ such that the
order on $\Lambda(c)$ is a linear extension of $\po$, the triangles
$i$ and $j$ of $\Tr(c)$ are disjoint.

\medskip

Our algorithm gradually refines a set of partial orders
on $\{0,1,\ttt_1, \ttt_2, \ldots, \ttt_{9n}\}$ with the constraint
that, at any time, every good order is a linear extension of at least
one of these partial orders. (Note that we do not need to make $\psi$
explicit.) Every partial order is refined until all or none of its
extensions are good, so that we can report success or discard that
partial order. Refinements are done in two ways:

\begin{itemize}
\item \emph{branching} over an uncomparable pair, meaning duplicating
  the partial order and adding the comparison in one copy, and its
  reverse in the other copy,
\item \emph{forcing} a comparison when it is required for the formula
  $S$ to be satisfiable.
\end{itemize}

\noindent
We keep our algorithm as simple as possible to facilitate the
verification of the algorithm, the code in Appendix~\ref{a:code}, and
the geometric results proven with their aid. This comes at the cost of
some efficiency, but we discuss some possible improvements in
Section~\ref{s:algodiscuss}.

\subsection{Description}

Our poset representation stores (i) for each lifted variable the
interval $(\cdot,0)$, $(0,1)$ or $(1,\cdot)$ that contains it, and
(ii) a directed graph over the variables contained in the interval
$(1, \cdot)$. The graph has $3n$ vertices, by Lemma~\ref{l:lift1}. To
compare two variables, we first retrieve the intervals containing
them. If they differ, we can return the comparison readily. If they
agree, then up to composing by $f$ or $f^{(2)}$ we can assume that
both variables are in $(1, \cdot)$ and we use the graph to reply.  We
ensure throughout that the graph is saturated, \emph{i.e.} is its own
transitive closure. In our implementation, initialization takes
$O(n^3)$ time, elements comparison takes $O(1)$ time, and edge
addition to the graph takes $O(n^2)$ time.

\bigskip

We start with the poset of the comparisons forced by the tagged
pattern: all pairs $(x_i,x_j)$, $\pth{f(x_i), f(x_j)}$, \ldots,
$\pth{f^{(2)}(z_i), f^{(2)}(z_j)}$ as well as pairs separated by $0$
or $1$. We next collect in a set $U$ the comparisons missing to
compute the vectors $v(i,j)$.

\begin{lemma}\label{l:U}
  $U$ contains only pairs of the form $z_k-f(y_k)$, $x_k-f(z_k)$, or
  $y_k-f(x_k)$.
\end{lemma}
\begin{proof}
  Every orientation predicate considered involves three points of the
  same index. Consider for instance $[X_i,Y_i,Z_i,X_j]$. Following
  Table~\ref{t:decomp}, this decomposes into
  $(x_i-x_j)(y_i-1)(z_i-f(y_i))$ and only the sign of the last term
  may be undecided. Other cases are similar and show that $U$ can only contain terms the form $z_k-f(y_k)$, $x_k-f(z_k)$, or
  $y_k-f(x_k)$.
\end{proof}

\medskip

\noindent
Every pair in $U$ corresponds to two variables \emph{with same index},
so $|U| \le 3n$. If $U$ contains the three pairs with a given index,
then two of the eight choices for these three comparisons are cyclic,
and can thus be ignored. We thus have at most $6^n$ ways to decide the
order of the undeterminate pairs of $U$; call them
\emph{candidates}. For each candidate, we make a separate copy of our
current graph and perform the following operations on that copy:

\begin{enumerate}

\item We add the $|U|$ edges ordering the undecided pairs as fixed by
  the candidate and compute its transitive closure. We check that the
  result is acyclic; if not, we discard that candidate (as it makes
  contradictory choices) and move to the next candidate.

  \medskip
  
\item Let $\po$ denote the resulting partial order. We consider every
  $1 \le i<j \le n$ in turn. (Note that $v(i,j)$ is determined and
  equal for all linear extensions of $\po$.)

  \begin{enumerate}[2a]
  \item If $v(i,j)_1 =v(i,j)_2 = v(i,j)_3$ or $v(i,j)_4 =v(i,j)_5 =
    v(i,j)_6$ then triangles $i$ and $j$ are disjoint in $\po$. We
    move on to the next pair $(i,j)$.
    
    \medskip

  \item Otherwise, the extensions of $\po$ in which the triangles $i$
    and $j$ are disjoint are those in which $[A_i,B_i,A_j,B_j] = 1$ or
    $[A_i,C_i,C_j,A_j] = 1$ (in the notations of
    Section~\ref{s:parameter}). Lemma~\ref{l:garcimore} asserts that
    $\po$ already determines at least one of these two predicates.

  \medskip

  \begin{enumerate}[2b1]

  \item If both tests are determined to false, then triangles $i$ and
    $j$ intersect in $\po$. We then discard $\po$ and move on to the
    next candidate.

    \medskip
  
  \item If one test is determined to false and the other is
    undetermined, then that second test must evaluate to true in every
    good extension of $\po$. Again, by Table~\ref{t:decomp} we are
    missing exactly one comparison to decide that test. We add it to
    our graph.

    \medskip

  \item In the remaining cases, at least one test is determined to
    true, so triangles $i$ and $j$ are disjoint in $\po$. We move on
    to the next pair $(i,j)$.
    
  \end{enumerate}
    
  \end{enumerate}

  \item If we exhaust all $(i,j)$ for a candidate, then we report ``realizable''.

  \medskip

  \item If we exhaust the candidates without reaching step $3$, then
    we report ``unrealizable''.

\end{enumerate}

\bigskip

\noindent
This algorithm relies on property whose computer-aided proof is discussed in Section~\ref{s:experiments}:

\begin{lemma}\label{l:garcimore}
  At step~$2b$, at least one of $[A_i,B_i,A_j,B_j]$ or
  $[A_i,C_i,C_j,A_j]$ is determined.
\end{lemma}

\subsection{Discussion}
\label{s:algodiscuss}

Let us make a few comments on our algorithm.

\subparagraph{Correctness.}

Let $\po_0$ denote the initial poset. First, remark that we explore
the candidates exhaustively, so every good extension of $\po_0$ is a
good extension of $\po_0$ augmented by (at least) one of the
candidates. Next, consider the poset $\po$ obtained in step 2. When
processing a pair $(i,j)$, we either discard $\po$ if we detect
that $i$ and $j$ intersects in it (2b1) or we move on to the next pair
$(i,j)$ after having checked (2a, 2b3) or ensured (2b2) that triangles
$i$ and $j$ are disjoint in $\po$. If we reach step~3, then all
extensions of the current partial order are good and we correctly
report feasibility. If a candidate is discarded then no linear
extension of $\po$ augmented by that candidate is a good order. If we reach Step~4, then every candidate has been discarded, so no linear
extension of $\po_0$ was a good order to begin with, and we correctly
report unfeasibility.

\subparagraph{Complexity.}

Initializing the poset and computing $U$ take $O(n^3)$ time. We have
at most $6^n$ candidates to consider. Step~1 takes $O(n^3)$ time. The
steps 2a-2b3 are executed $O\pth{n^2}$ times, and the bottleneck among them is 2b2, which takes $O(n^2)$ time. Altogether, our algorithm decides if a tagged pattern is realizable in $O\pth{6^nn^4}$ time.

\subparagraph{Improvements.}

In practice, the algorithm we presented can be sped up in several
ways. For example, it is much better to branch over the pairs of $U$
one by one. Once a branching is done, we can update $U$ by removing
the pairs that have become comparable, and thus avoid examining
candidates that would get discarded at Step~1. Also, it pays off to
record the forbidden tagged patterns of small size, and, given a
larger tagged pattern to test, check first that it does not contain a
small forbidden pattern.

\subparagraph{One-sided certificate.}

If the algorithm reaches Step~3, we actually know a poset for which
every linear extension is good. This means that we can compute an
arbitrary linear extension to obtain an order on the variables in
$(1,\cdot)$. We can then assign to these variables any values that
satisfy this order, say by choosing the integers from $2$ to $3n+1$,
and then propagate these values via $f$ and $f^{(2)}$ to all lifted
variables. From there, we can extract the values of $x_1,x_2, \ldots,
z_n$ of a concrete realization of our tagged pattern. In this way, all
computations are done on (relatively small) rationals and are
therefore easy to do exactly.

\section{Experimental results}
\label{s:experiments}

We now discuss our implementation of the above algorithm as well as
its experimental use. Remember that we call a tagged pattern forbidden
if it admits no \emph{canonical} realization. We make the raw data
available (see Appendix~\ref{a:raw}).

\subparagraph{Implementation.}

We implemented the algorithm of Section~\ref{s:algorithm} in Python 3, and comment on its key functions in Appendix~\ref{a:code}. For simplicity, our implementation makes one  adjustment to the algorithm: we branch over all $2^{|U|}$ choices for the pairs of undecided variables; so, we take $8$ choice per $k$, rather than $6$. Altogether, the implementation amounts to $\sim 470$  lines of (commented) code and is sufficiently effective for our experiments: on a standard desktop computer, finding all realizable triples of size~$6$ (and a realization when it exists) takes about $40$ minutes, whereas verifying that no minimally forbidden tagged pattern of size $6$ exists took up about a month of computer time; the difference of course is that in the former, for realizable triples we do not have to look at all positions of tags.

\subparagraph{Proof of Lemma~\ref{l:garcimore}.}

The statement concerns only two triangles and can be shown by a simple
case analysis. Our code sets up an exception that
is raised if the statement of the lemma fails (cf line 94 in the code in Appendix~\ref{a:exception}). Checking the realizability of all tagged patterns on two elements exhausts the case analysis, and the exception is not raised.

\subparagraph{Minimally forbidden patterns.}

To state the minimally forbidden tagged patterns of size~$3$ we
compress the notation as follows. We use $\{uv\}$ to mean ``$uv$ or
$vu$''. Symbols that are omitted may be placed anywhere (this may
include $\0$ and $\1$). We use $x_i = y_j$ to mean ``any pattern in
which the $i$th symbol on the $1$st tagged permutation equals the
$j$th symbol of the $2$nd tagged permutation''.

\begin{proposition}\label{p:tag}
  The equivalence classes of minimally forbidden tagged patterns are:

  \begin{enumerate}[(i)]
  \item For size $2$, $(ab\0,\1ab,ab)$, $(\0ab,ab\1,ba)$,
    $(ab\0,ba\1,ab)$, and $(\0ab,\1ba,ba)$.

    \medskip
    
  \item For size $3$, $(\{ab\}\0, \{ab\}c\0, z_2 = y_1)$, $(\{ab\}c\0,
    \1\{ab\}, z_2 = x_1)$, $(\{ab\}\0, \1c\{ab\}, z_2 = y_3)$,
    $(\1c\{ab\}, \1\{ab\}, z_2 = x_3)$, $(abc\0, b\1ac,ca\0b)$, and
    $(\1abc,b\1ca,ac\0b)$.

    \medskip

  \item For size $4$, the taggings of $(abcd, badc, cdab)$ that
    contains $(\0b\1, \1d, a\0)$ or $(b\0c, \1a, a\0)$, and the
    taggings of $(abcd, badc, dcba)$ that contains $(b\0c,\0d\1,\1c)$
    or $(c\0,\{\0ba\}\{\1dc\},\1c)$.

    \medskip

  \item None for size $5$ and $6$.
  \end{enumerate}
\end{proposition}

\subparagraph{Realization database.}

For every tagged pattern that our algorithm declared
realizable, we computed a realization (as explained in
Section~\ref{s:algodiscuss}) and checked it independently.

\subparagraph{Geometric permutations. }It remains to prove our statements on geometric
permutations:

\begin{proof}[Proof of Theorem~\ref{thm:five}]
  For every triple of permutations, we checked that it is realizable
  by trying all $8$ reversals and all $\binom72^3$ possible positions
  of $\0$ and $\1$, until we find a choice that does not contain any
  minimally forbidden tagged pattern of Proposition~\ref{p:tag}.
\end{proof}

\begin{proof}[Proof of Theorem~\ref{thm:so}]
  We argue by contradiction. Consider four disjoint convex sets met by
  lines $\ell_1$, $\ell_2$ in the order $abcd$ and $\ell_3$ in the
  order $badc$; assume that $\ell_1 \odot \ell_3 = - \ell_2 \odot
  \ell_3$. By the perturbation argument of Lemma~\ref{l:canonical}, we
  can assume that the three lines are pairwise skew and that the
  convex sets are triangles with vertices on these lines. Moreover,
  there exists a nonsingular affine transform $A$ that maps the
  unoriented lines $\ell_1$ to $\ell_x$, $\ell_2$ to $\ell_y$ and
  $\ell_3$ to $\ell_z$. Remark that $A$ either preserves or reverses
  all side operators. Since $\ell_x \odot \ell_z = \ell_y \odot
  \ell_z$ and $\ell_1 \odot \ell_3 = - \ell_2 \odot \ell_3$, the map
  $A$ sends the oriented lines $(\ell_1,\ell_2)$ to either $(\ell_x,
  -\ell_y)$ or $(-\ell_x, \ell_y)$. We used our program to check that
  none of $(abcd, dcba, badc)$, $(dcba, abcd, badc)$, $(abcd, dcba,
  cdab)$, $(dcba, abcd, cdab)$ admits a canonical realization. The
  statement follows.
\end{proof}

\section{Geometric analysis}
\label{s:analysis}

We present here an independent proof that the tagged patterns of
size~$2$ and~$3$ listed in Proposition~\ref{p:tag} do not have a
(canonical) realization. We do not prove the patterns are minimal, nor
do we prove that the list is exhaustive; these facts come from the
completeness of our computer-aided enumeration.

\subsection{Size two}

The following observation was used by Asinowski and
Katchalski~\cite{Asinowski2005}:

\begin{obs}\label{AK-obs}
  Let $X$ and $Y$ be compact convex sets and let $P$ and $Q$ be points
  in $\mathbb{R}^3$. Assume that $X,Y,P,Q$ are pairwise disjoint and
  that there exist lines  inducing the geometric permutation $(PXY)$
  and $(QYX)$. Any oriented line with direction $\vc{PQ}$ that
  intersects $X$ and $Y$, must intersect $X$ before $Y$.
\end{obs}

\begin{figure}[!ht]
  \begin{center}\includegraphics[page=5,width = 7cm]{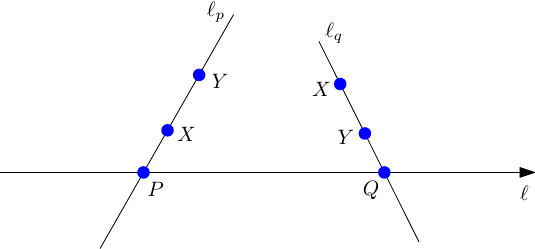}\end{center}
\end{figure}

\begin{proof}
  Refer to the figure. Let $h$ be a plane that separates $X$ and
  $Y$. The existence of the geometric permutations $(PXY)$ and $(QYX)$
  ensure that $h$ also separates $P$ and $Q$. Moreover, the halfspace
  bounded by $h$ that contains $X$ also contains $P$, so any line with
  direction $\vc{PQ}$ traverses $h$ from the side of $X$ to the side
  of $Y$.
\end{proof}

Observation~\ref{AK-obs} implies that $(ab\0,\1ab,ab)$,
$(\0ab,ab\1,ba)$, $(ab\0,ba\1,ab)$, and $(\0ab,\1ba,ba),$ are
forbidden. Indeed, consider, by contradiction, a realization of one of
these tagged patterns. Let $P$ be the point $\0$ on $\ell_x$ and $Q$
the point $\1$ on $\ell_y$. In each case, we can map $X$ and $Y$ to
$a$ and $b$ so that some line $\ell_P \in\{ \ell_x, -\ell_x\}$
realizes $PXY$ and $\ell_Q \in \{\ell_y, -\ell_y\}$ realizes
$QYX$. Then, Observation~\ref{AK-obs} implies that any line with same
direction as the line from $P$ to $Q$ must intersect $X$ before $Y$;
this applies to the line $\ell_z$ and contradicts the fact that the
configuration realizes the chosen tagged pattern.

\subsection{Size three}

To argue that the tagged patterns of size~$3$ of
Proposition~\ref{p:tag} are forbidden we first need a basic
observation concerning planar geometric permutations.

\begin{obs} \label{planar-obs}
  Suppose the $x$-axis and the $y$-axis are transversal to three
  disjoint convex sets in $\mathbb{R}^2$. Suppose furthermore that all
  the sets intersect the $x$-axis in points $x_1, x_2, x_3$, where
  either $x_1 < x_2< x_3\leq 0$ or $0\leq x_1< x_2< x_3$. Then the
  middle element of the geometric permutation induced by the $y$-axis
  can not equal the extreme element of the $x$-axis. (Here the extreme
  element refers to the set intersected the farthest away from the
  origin on the $x$-axis.)
\end{obs}
\begin{proof}
  Lets call the sets $A_1$, $A_2$, $A_3$.  Up to symmetry we may
  assume that the sets intersect the $x$-axis in the points $0\leq
  x_1<x_2<x_3$ where $x_i\in A_i$. This means that $A_3$ is the
  extreme element of the $x$-axis. Now suppose for contradiction that
  $A_3$ is the middle element of the geometric permutation determined
  by the $y$-axis, so the $y$-axis meets the sets in points $y_i < y_3
  < y_j$, with $y_k \in A_k$.  If $0<y_3$, then the segment
  $[x_j,y_j]$ intersects the segment $[x_3,y_3]$, and if $y_3<0$, then
  the segment $[x_i,y_i]$ intersects the segment $[x_3,y_3]$.
\end{proof} 

\noindent
Now, Proposition~\ref{p:tag}~(ii) asserts that the patterns
\[\begin{array}{ll}
1. \quad (\{ab\}\0, \{ab\}c\0, z_2 = y_1), \quad & 
4. \quad (\1c\{ab\}, \1\{ab\}, z_2 = x_3), \quad \\
2. \quad (\{ab\}c\0, \1\{ab\},z_2 = x_1), \quad & 
5. \quad (abc\0, b\1ac,ca\0b), \quad \\
3. \quad (\{ab\}\0, \1c\{ab\}, z_2 = y_3),  & 
6. \quad (\1abc,b\1ca,ac\0b).
\end{array}\]
are forbidden. The basic idea of the proof of this fact is to show under the given conditions we can find another
transversal line that intersects one of the lines $\ell_x$, $\ell_y$,
$\ell_z$ to obtain a pair of crossing lines where we can apply
Observation \ref{planar-obs}.

\bigskip

For patterns 1--4, the line we are looking for is just a translate of
the line $\ell_z$. The translation can be found by considering an
appropriate projection. In Figure~\ref{fig:observation 6}  we see projections of
patterns 1--4 to the $xy$-plane.

\begin{figure}[t]
  \begin{center}\includegraphics[page=6, width = 14cm]{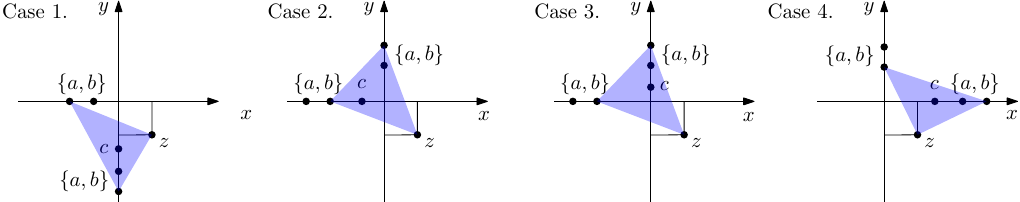}\end{center}
  \caption{Projections to the $xy$-plane of the patterns 1--4. In each case the projection of triangles $a$ and $b$ contain the point labeled by $c$. \label{fig:observation 6}}
\end{figure}

\noindent
Specifically:
\begin{itemize}
\item In cases 1 and 3 we see that the line $\ell'$ parallel to
  $\ell_z$ which passes through the point $c_y$ (the point of $c$ on
  the line $\ell_y$) is transversal to all the sets. Therefore we can
  apply Observation \ref{planar-obs} with $\ell_y$ as the $x$-axis and
  $\ell'$ as the $y$-axis.
\item In cases 2 and 4 we see that the line $\ell'$ parallel to
  $\ell_z$ which passes through the point $c_x$ is transversal to all
  the sets. Therefore we can apply Observation \ref{planar-obs} with
  $\ell_x$ as the $x$-axis and $\ell'$ as the $y$-axis.
\end{itemize}

\bigskip

For cases 5 and 6 we will need one additional observation:

\begin{obs}\label{boloid}
  Consider two lines $\ell_1$ and $\ell_2$ in $\R^3$ and three
  segments $a$, $b$ and $c$, each with one endpoint on $\ell_1$ and
  one endpoint on $\ell_2$. Assume that $a$, $b$ and $c$ are pairwise
  non-coplanar and not all three parallel to a common plane. Put $a_i
  = a \cap \ell_i$ and parameterize the segment $a_1a_2$ as
  $a(t) = (1-t) a_1 + t a_2$, $t \in [0,1]$. Let $\ell(t)$ denote
  the line through $a(t)$ that intersects or is parallel to the lines
  supporting respectively $b$ and $c$. The lines $\ell_1$ and $\ell_2$
  realize the same geometric permutation of $\{a,b,c\}$ if and only if
  for every $t \in [0,1]$ the line $\ell(t)$ intersects the three
  segments $a$, $b$ and $c$.
\end{obs}
\begin{proof}
  Let us first reformulate the statement. For $s \in \{a,b,c\}$ let
  $\ell_s$ denote the line supporting $s$. The intersection point
  $\ell(t) \cap \ell_b$ moves along $\ell_b$ from $\ell_1 \cap \ell_b$
  to $\ell_2 \cap \ell_b$ by travelling either along the segment $b$,
  or along $\ell_b \setminus b$. We are in the latter case if and only
  if for some $0<t<1$ the line $\ell(t)$ is parallel to $\ell_b$. The
  same holds for $c$. The statement therefore reformulates as:
  $\ell_1$ and $\ell_2$ realize the same geometric permutation of
  $\{a,b,c\}$ if and only if no line $\ell(t)$ is parallel to $\ell_b$
  or $\ell_c$. Now, the reformulated statement follows from elementary
  considerations on the geometry of ruled quadric surfaces (see for instance
  the projective geometry textbook of Veblen and
  Young~\cite[Chapter 11]{VebYou10}). Let us spell it out.

  Since $\ell_a$, $\ell_b$ and $\ell_c$ are in general position, these
  lines and all their common transversals are contained in a quadric
  surface $H$, specifically a hyperbolic paraboloid. The quadric $H$
  has two families of rulings, one which contains $\{\ell_a, \ell_b,
  \ell_c\}$ and the other that we denote as $R$. Every line in one family
  of rulings intersects or is parallel to every line in the other
  family of rulings.
  
  Consider a space parameterizing the lines of $\R^3$ bicontinuously
  (for instance via Pl\"ucker coordinates) and identify a line with
  its parameter point. In that space, $R$ forms a closed loop and it
  contains five points of interest to us: $\ell_1$ and $\ell_2$, as
  well as three special lines, which are transversals to two of
  $\{\ell_a, \ell_b, \ell_c\}$ and parallel to the third one. Let us
  denote the special lines by $\ell_a^*$, $\ell_b^*$ and $\ell_c^*$,
  where $\ell_s^*$ is parallel to $\ell_s$.

  The special lines $\ell_a^*$, $\ell_b^*$ and $\ell_c^*$ therefore
  split $R$ into three open arcs, and two line transversals to
  $\{\ell_a, \ell_b, \ell_c\}$ realize the same (untagged) geometric permutation
  if and only if these lines belong to the same arc. Indeed, as a line
  $\ell$ moves continuously on $R$, its intersection points with
  $\ell_a$, $ \ell_b$ and $\ell_c$ change continuously -- and in
  particular the order in which the line meets $\ell_a$, $ \ell_b$ and
  $\ell_c$ remains unchanged -- \emph{except} as the $\ell$ reaches
  one of the special line $\ell_s^*$: when that happens, the
  intersection with $\ell_s$ jumps from one ``end'' of $\ell$ to its
  other ``end''. Formally, as $\ell$ passes one of $\ell_a^*$,
  $\ell_b^*$ and $\ell_c^*$, the order in which the moving line
  intersects $\ell_a$, $ \ell_b$ and $\ell_c$ changes by a circular
  permutation.

  Now consider the set of lines $\gamma = \{\ell(t) \colon 0 \le t \le
  1\}$. This is an arc contained in $R$, bounded by $\ell_1$ and
  $\ell_2$ and, by assumption, not containing $\ell_a^*$. It follows
  that $\ell_1$ and $\ell_2$ realize the same geometric permutation of
  $\{a,b,c\}$ if and only if $\gamma$ contains neither $\ell_b^*$ nor
  $\ell_c^*$. This is precisely the reformulation of the statement.
\end{proof}

Let us now get back to proving that the tagged patterns 5 and 6 are
forbidden. Notice that in both cases, the directed lines $\ell_y$ and
$-\ell_z$ induce the same permutations. Thus, by
Observation~\ref{boloid} there is a continuous family of directed
lines $\{\ell(t) : 0\leq t\leq 1\}$ that are all transversal to the
segments $a$, $b$ and $c$ in the same order; moreover, $\ell(0) =
\ell_y$ and $\ell(1) = -\ell_z$. Next notice that $\ell_y$ and
$-\ell_z$ lie on opposite sides of the line $\ell_x$ (that is,
$\ell_y\odot\ell_x \neq (-\ell_z )\odot \ell_x$), so by continuity
there exists $0<t'<1$ such that the directed line $\ell' = \ell(t')$
intersects the line $\ell_x$. We now apply Observation
\ref{planar-obs} to the pair of lines $\ell_x$ and $\ell'$.

We start with Case 5. The line $\ell'$ intersects the segments in the
order $b\prec a \prec c$, and we will consider where the intersection
point $P = \ell_x\cap \ell'$ fits into this ordering. If the order of
intersection along $\ell'$ is $b\prec a \prec c \prec P$, then we can
apply Observation \ref{planar-obs} with $\ell'$ as the $x$-axis and
$\ell_x$ as the $y$-axis. We may therefore assume that $P\prec c$. If
we in addition have $b \prec P$, then it follows that the
$x$-coordinate of $P$ is non-negative. To see this observe that the
$x$-coordinates of the points on segments $[b_y,b_z]$ and $[c_y,c_z]$
are between $0$ and $1$, so the same holds for any point on the line
$\ell'$ which is between $b$ and $c$, in particular for the point
$P$. Since the points $a_x,b_x,c_x$ all have negative $x$-coordinates
it follows that the line $\ell_x$ intersects the points in the order
$a\prec b\prec c\prec P$, and we may therefore apply Observation
\ref{planar-obs} with $\ell_x$ as the $x$-axis and $\ell'$ as the
$y$-axis. The final situation to be considered is if $\ell'$
intersects the points in the order $P\prec b \prec a\prec c$. However
the plane $\{y + z = 1\}$ strictly separates the segment $[b_y,b_z]$
from $[a_y,a_z]\cup [c_y,c_z]\cup \ell_x$, and so if the order along
$\ell'$ was $P\prec b \prec a\prec c$, this would force $\ell'$ to
intersect the plane $\{y = z+1\}$ twice, a contradiction.

Now for Case 6 which is similar. The line $\ell'$ intersects the segments in the order $b\prec c\prec a$ and we want to place the intersection point $P=\ell_x\cap \ell'$ in this ordering. If we have $b\prec c \prec a \prec P$ then we can apply Observation \ref{planar-obs} with $\ell'$ as the $x$-axis and $\ell_x$ as the $y$-axis. If we have $b\prec P \prec a$, then the $x$-coordinate of the point $P$ is at most $1$ (since the $x$-coordinates of the points on the segments $[b_y,b_z]$ and $[a_y,a_z]$ are between $0$ and $1$). We may therefore apply Observation \ref{planar-obs} with $\ell_x$ as the $x$-axis and $\ell'$ as the $y$-axis. The final situation is the order $P\prec b\prec c\prec a$ which is impossible since the plane $\{y +z = 1\}$ separates the segment $[b_y,b_z]$ from $[a_y,a_z]\cup [c_y,c_z] \cup \ell_x$.

\bibliographystyle{plainurl}
\bibliography{ref}

\appendix

\section{Additional material}
\label{a:raw}
    
We provide the following additional material:    

\begin{itemize}
    \item The python code. See the \texttt{README.txt} file for how to invoke it.
    \medskip
    \item Some output files summarizing the results of our experiments:
    \begin{itemize}
        \item The list of triples of permutations of size five with, for each, the coordinates of a geometric realization. The triples are normalized in the sense defined in Appendix~\ref{a:forbid6}.
        \item The list of triples of size six with, for each, either the coordinates of a geometric realization or the statement that none exists.
    \end{itemize}
\end{itemize}
    
\noindent We make this material available from

\begin{quote}
  \url{https://hal.inria.fr/hal-02050539/file/Code.tar.gz}.
\end{quote}

\section{Code}
\label{a:code}

We present here the implementation that we used for our experiments. We do not reproduce the entire code as it is not so informative, but explain how its key functions operate. Readers interested in the full code can get it from the url given in Appendix~\ref{a:raw}.

\subsection{\texttt{compute\_realization\_tagged}}
\label{a:exception}

The function \texttt{compute\_realization\_tagged} is the main algorithm. It checks whether a given triplet of tagged permutations is canonically realizable, and output a realization if the answer is positive. 

\subparagraph{Encoding.}

Permutations of $\{0,\ldots,n-1\}$ are encoded as lists. Tagged are encoded separately as pairs $[z,o]$ where $z$ (resp. $o$) is the position of the zero (resp. of the one), defined as the number of elements smaller than zero (resp. one). For instance $23|^{\bf 0}0|^{\bf 1}1$ is encoded with the permutation $[2,3,0,1]$ and the tags $[2,3]$.

Our algorithm first consists of collecting the information that is missing to compute the sign vectors of Guigue-Devillers' algorithm in a set $U$. As stated in Lemma~\ref{l:U}, it only contains specific pairs, such as $x_k-f(z_k)$, which are encoded by the number associated with~$x_k$. The numbers of $x_i$, $y_j$ and $z_k$ are $i$, $n+j$ and $2n+k$, respectively.

In $M$ we keep the transitive closure of the adjacency matrix of the graph that collects the orientations induced by the tagged permutations and by the current orientations of the elements in $U$. This graph represents the poset of the relative orders of the variables greater than one: if $x_2$ is between $0$ and $1$, the associated vertex number $2$ represents $f(x_2)$ and if $y_1$ is smaller than $0$, the associated vertex $n+1$ represents $f^2(y_1)$.

\subparagraph{Auxiliary functions.}

The function \texttt{compute\_realization\_tagged} uses the following simple algorithms, for which we do not give the code here:

\begin{itemize}
    \item \texttt{associated\_vertex(v)} is the only variable $w$ such that
    $v-f(w)$ can appear in $U$.
    \item \texttt{compute\_base\_graph} build the adjacency matrix of the graph with an edge $v\rightarrow w$ whenever $v$ and $w$ are on the same line, in the same interval, and $v$ is before $w$ in the associated permutation.
    \item \texttt{transitive\_acyclic(M)} add an edge $v\rightarrow w$ whenever there is a non-trivial path from $v$ to $w$, using Warshall algorithm. Halts and return false if $M$ is not acyclic.
    \item \texttt{add\_edge\_closure(M,v,w)} add an edge $v\rightarrow w$ and update the transitive closure.
    \item \texttt{topological\_sort(M)} compute an ordering of the vertices compatible with the graph, using the classical depth-first algorithm.
\end{itemize}

\subparagraph{Code of the function.}

~ \bigskip

\begin{lstlisting}[language=Python]
def compute_realization_tagged(perm_triplet, tag_triplet):
    Px, Py, Pz = perm_triplet
    n = len(Px)
    Ix, Iy, Iz = inverse_perm(Px), inverse_perm(Py), inverse_perm(Pz)
    inv_triplet = [Ix, Iy, Iz]

    # precompute the interval of each vertex: 
    # 0 if <0, 1 if between 0 & 1, and 2 if >1
    regions = [compute_region(v, inv_triplet, tag_triplet) 
               for v in range(3*n)]

    # initialization of the adjacency matrix of the graph 
    M = [[False]*(3*n) for _ in range(3*n)]

    # the base matrix containing information from the perms and tags
    B = compute_base_graph(perm_triplet, tag_triplet)

    # SV is a map, s.t. SV[(i,j)] is the sign vector of (i,j),
    # which is a 6-tuple containing pairs of the form (sign, unknown):
    # sign = +- 1 and unknown = None if the sign is determined;
    # otherwise, unknown = v if the value is the sign of sign*(v-f(w)) 
    # U is the set of unknown and required v-f(w), 
    # identified by their v (w is associated_vertex(v))
    SV, U = {}, set([])
    for i, j in combinations(range(n), 2):  
        SV[(i, j)] = list()
        for l in range(3):
            SV[(i, j)].append(compute_initial_orientation(i, l, j, 
                inv_triplet, regions, U))
        for l in range(3):
            SV[(i, j)].append(compute_initial_orientation(j, l, i,
                inv_triplet, regions, U))
    m = len(U)

    # main loop: iteration over the orientations of the elements of U
    for it in range(2**m):  
        # copy the base B into M
        for ii, jj in product(range(3*n),range(3*n)):  
            M[ii][jj] = B[ii][jj]
        # if the i-th bit of "it" is 1 then v < f(w) for the i-th v 
        # of U and v > f(w) otherwise
        j = it
        for v in U:
            w = associated_vertex(v, n)
            if j % 2 == 1:
                M[v][w] = True
            else:
                M[w][v] = True
            j //= 2  # continue the bit decomposition
        if not transitive_acyclic(M):
            # the graph is not acyclic, we try the next choice of
            # orientations for the elements in U
            continue
        for i, j in combinations(range(n), 2):  
            # check whether they are pairwise disjoint
            intersection = False # used to halt the loop if needed
            # compute the sign vector of (i,j), now that 
            # we have all the required information stored in M
            sv = compute_sign_vector(SV[(i, j)], M)

            # we now apply the Devillers-Guigue algorithm
            if sv[0] == sv[1] == sv[2] or sv[3] == sv[4] == sv[5]:
                # the two triangles are disjoint, go to the next pair
                continue
            # compute the numbers associated with the variables
            Xi, Yi, Zi, Xj, Yj, Zj = i, i+n, i+2*n, j, j+n, j+2*n
            # order the points as in the Devillers-Guigue's algorithm
            sign_i, sign_j = sv[0], sv[3]
            if sv[0] == sv[1]:
                Xj, Yj, Zj = Zj, Xj, Yj
                sign_i = sv[2]
            elif sv[0] == sv[2]:
                Xj, Yj, Zj = Yj, Zj, Xj
                sign_i = sv[1]
            if sv[3] == sv[4]:
                Xi, Yi, Zi = Zi, Xi, Yi
                sign_j = sv[5]
            elif sv[3] == sv[5]:
                Xi, Yi, Zi = Yi, Zi, Xi
                sign_j = sv[4]
            if sign_i == -1:
                Yi, Zi = Zi, Yi
            if sign_j == -1:
                Yj, Zj = Zj, Yj
            # the points have been permuted so that their sign vector 
            # is now (1,-1,-1,1,-1,-1). We (try to) compute the two 
            # orientations of the algorithm:
            # each O is either sign, None if determined
            # or sign, [v,w] if it has the sign of sign*(v-f(w))
            O1 = compute_last([Xi,Yi,Xj,Yj], M, inv_triplet, regions)
            O2 = compute_last([Xi,Zi,Zj,Xj], M, inv_triplet, regions)
            if O1[1] is not None and O2[1] is not None:  
                # for some reason, this case never happen
                raise Exception("Both orientations are unknown")
            if O1[1] is None and O1[0] == 1:
                continue  # the two triangles are disjoint
            if O2[1] is None and O2[0] == 1:
                continue  # the two triangles are disjoint
            if O1[1] is None and O2[1] is None:  
                # they are both -1 at this point => intersection
                intersection = True
                break  # the (i,j) loop
            if O1[1] is None: 
                # O1 is -1 and O2 is unknown: it has to be 1 
                v, w = O2[1]
                if O2[0] == 1:
                    add_edge_closure(M, w, v)  # v - f(w) > 0
                else:
                    add_edge_closure(M, v, w)  # v - f(w) < 0
            if O2[1] is None: 
                # O2 is -1 and O1 is unknown: it has to be 1
                v, w = O1[1]
                if O1[0] == 1:
                    add_edge_closure(M, w, v)  # v - f(w) > 0
                else:
                    add_edge_closure(M, v, w)  # v - f(w) < 0
        if intersection:
            # a break occurred because of a detected intersection
            continue
        # at this stage a realization is possible
        L = topological_sort(M)
        return output_realization(L, regions, n)
    return None
\end{lstlisting}
\bigskip

Notice line 94 of the code, which sets up an exception. The fact that it is never raised as we examine all tagged patterns on two elements is what proves Lemma~\ref{l:garcimore}.

\subsection{\texttt{compute\_initial\_orientation} and \texttt{compute\_last}}

These two functions are used to compute orientations. Note that, because of the very specific underlying determinants, we compute them "by hand" instead of relying on computer algebra system.

\subparagraph{First step (vectors $v(i,j)$).}

The first one is called while building the sign vectors, collecting unknowns in $U$. We are therefore only interested in orientations of $(X_i,Y_i,Z_i,L_j)$, which factorizes as $(z_i-z_j)(x_i-1)(y_i-f(x_i))$ or alike. The sign of the first factor is given by a permutation, the one of the second factor by the tags, so it only remains to compare $y_i$ to $f(z_i)$: if they are not in the same interval it is determined, otherwise we add it in $U$.

\bigskip

\begin{lstlisting}[language=Python]
def compute_initial_orientation(i, l, j, inv_triplet, regions, U):
    """Computes the orientation of one entry (Xi, Yi, Zi, lj).
    If it is determined by the tagged permutations, returns sign, None 
    Otherwise, it is of the form sign*sign_of(v-f(w)),with sign=+-1 
    and it returns sign, v;
    v is then added to the set U which collects the missing information"""
    n = len(inv_triplet[0])
    s, t, u = l, (l+1) % 3, (l+2) % 3  # the three axis in direct order
    if inv_triplet[s][i] > inv_triplet[s][j]:  # (si - sj) > 0
        sign = 1
    else:
        sign = -1
    if regions[t*n+i] < 2:  # (ti - 1) < 0
        sign = -sign
    v, w = i + u*n, i + t*n  # encode of ui and ti as numbers
    r1, r2 = regions[v], (regions[w] + 1) % 3  # regions of v and f(w)
    if r1 == r2:  # v and f(w) are in the same region, sign unknown
        U.add(v)
        return sign, v
    if r1 > r2:  # ui-f(ti) > 0
        return sign, None
    return -sign, None
\end{lstlisting}

\subparagraph{Second step (relabelling and final condition).}

The second one is more complicated, since it is called when computing the last two orientations of Devillers-Guigue's algorithm and any situation can occur. So we first sort the points (while computing the number of inversions performed to have a correct sign), then refer to Table~\ref{t:decomp}. Once again, either we have the required information directly or in $M$, or we output the missing $v-f(w)$ as a pair $(v,w)$. Note that the code is quite long on purpose, to stick with Table~\ref{t:decomp}: a more factorized code would be more difficult to read.

\bigskip

\begin{lstlisting}[language=Python]
def compute_last(O, M, inv_trip, regions):
    """Try to compute the orientation of O using M
    returns sign, None if successful where sign is the orientation
    returns sign, [v,w] otherwise, when the orientation is sign*(v-f(w))"""
    n = len(M) // 3
    # we first the variables while keeping (-1)^swaps in the variable sign
    # this way, we are in one of the cases described in the article
    sign = insertion_sort(O)
    line, index = [x // n for x in O], [x % n for x in O]

    if line[0] == line[1] and line[2] == line[3]:  
        # case (xa-xb)(yc-yd) or alike
        if inv_trip[line[0]][index[0]] < inv_trip[line[1]][index[1]]:
            sign = -sign  # (xa-xb) < 0
        if inv_trip[line[2]][index[2]] < inv_trip[line[3]][index[3]]:
            sign = -sign  # (yc-yd) < 0
        return sign, None
    if line[0] == line[1]:  # case [Xa,Xb,Yc,Zd]
        if inv_trip[line[0]][index[0]] < inv_trip[line[1]][index[1]]:
            sign = -sign  # (Xa-Xb) < 0
        if regions[O[2]] < 2:
            sign = -sign  # (Yc-1) < 0
        r3 = regions[O[3]]
        r2 = (regions[O[2]]+1) % 3
        if r3 < r2 or M[O[3]][O[2]]: # Zd < f(Yc)
            return -sign, None
        elif r3 > r2 or M[O[2]][O[3]]:  # Zd > f(Yc)
            return sign, None
        else:
            return sign, [O[3], O[2]]
    if line[1] == line[2]:  # case [Xa,Yb,Yc,Zd]
        sign = -sign
        if inv_trip[line[1]][index[1]] < inv_trip[line[2]][index[2]]:
            sign = -sign  # (yb-yc) < 0
        if regions[O[3]] < 2:
            sign = -sign  # (zd-1) < 0
        r0 = regions[O[0]]
        r3 = (regions[O[3]]+1) % 3
        if r0 < r3 or M[O[0]][O[3]]: # Xa < f(Zd)
            return -sign, None
        elif r0 > r3 or M[O[3]][O[0]]:  # Xa > f(Zd)
            return sign, None
        else:
            return sign, [O[0], O[3]]
    if line[2] == line[3]:  # case [Xa,Yb,Zc,Zd]
        if inv_trip[line[2]][index[2]] < inv_trip[line[3]][index[3]]:
            sign = -sign  # (zc-zd) < 0
        if regions[O[0]] < 2:
            sign = -sign  # (xa-1) < 0
        r1 = regions[O[1]]
        r0 = (regions[O[0]]+1) % 3
        if r1 < r0 or M[O[1]][O[0]]: # Yb < f(Xa)
            return -sign, None
        elif r1 > r0 or M[O[0]][O[1]]:  # Yb > f(Xa)
            return sign, None
        else:
            return sign, [O[1], O[0]]
\end{lstlisting}

\subsection{\texttt{output\_realization}}

We conclude this presentation of the main parts of the code with function 
\texttt{output\_realization}, which computes a realization. It simply ranges through the vertices in the order given by the topological sort, and assigns a value from $2$ to $3n+2$ to the associated variable known to be larger than $1$: it is then brought back to the correct interval by applying $f$ or $f^2$ if needed.
\bigskip

\begin{lstlisting}[language=Python]
def output_realization(L, regions, n):
    """Computes coordinates that match the order L and the  regions
    returns an array of strings C, where C[v] is the value of v on its line (as a fraction)"""
    C = [None] * (3 * n)
    i = 2
    for v in L:
        if regions[v] == 2:
            C[v] = Fraction(i)
        elif regions[v] == 0:
            C[v] = Fraction(-1, i - 1)
        else:
            C[v] = Fraction(i - 1, i)
        i += 1
    return C
\end{lstlisting}

\section{Forbidden triples of size six}
\label{a:forbid6}

We finally present \emph{in a normalized form} the complete list of triples of permutations of size 6 that are not geometrically realizable in $\R^3$. 

\bigskip

The normalization goes as follows. Recall that we write permutations as words, so we can order them using the lexicographic order on the associated words. Starting from a triple $(P_1,P_2,P_3)$ of geometric permutations, consider the six triples obtained by the following method:

\begin{itemize}
    \item Choose $P_1$, $P_2$, $P_3$ or one of their reverse as first permutation.
    \item Relabel so that this first permutation becomes the identity, and propagate the relabeling to the other two permutations and their reverses.
    \item Pick as second permutation the lexicographically smallest one among the two remaining (relabelled) permutations and their (relabelled) reverses.
    \item Pick as third permutation the lexicographically smallest one among the remaining (relabelled) permutation and its (relabelled) reverse.
\end{itemize}

\noindent
The normalization of $(P_1,P_2,P_3)$ is, among these six triples, the one that is lexicographically smallest (smallest first permutation, then among tied, smallest second permutation, etc.).

\bigskip

The next table lists the normalized triples of permutations that are not geometrically realizable in $\R^3$.

{\small
\begin{table}[!h]
\begin{center}
\begin{tabular}{|c|c|c|c|c|c|}
\hline
013524 104523&
013524 140523&
013524 325014&
013542 104523&
013542 140523&
013542 325014\\

014352 105243&
014523 102534&
014523 103425&
014523 103524&
014523 120534&
014523 130524\\

014523 210534&
014523 243015&
014532 102534&
014532 103524&
014532 105243&
014532 105342\\

014532 120534&
014532 130524&
014532 210534&
014532 310524&
015342 104253&
021453 102534\\

021453 120534&
023514 210534&
023514 210543&
023541 210534&
023541 210543&
024135 201453\\

024135 210543&
024135 215043&
024153 120534&
024153 210543&
024153 215043&
024315 201543\\

024315 210453&
024315 210543&
024315 215043&
024351 210453&
024351 210543&
024351 215043\\

024513 120534&
024513 210534&
024513 210543&
024513 215043&
024531 120534&
024531 210534\\

024531 210543&
024531 215043&
025134 201453&
025134 201543&
025134 204153&
025134 210453\\

025314 201453&
025314 204153&
025314 204513&
025314 210453&
025314 210543&
025314 214053\\
\hline
025314 214503&
025341 201453&
025341 204153&
025341 204513&
025341 210453&
025341 210543\\

025341 214053&
025341 214503&
032451 230154&
032451 230514&
032514 210534&
032514 210543\\

032514 234015&
032541 210534&
032541 210543&
034125 103254&
034125 104523&
034125 130524\\

034125 135024&
034125 140523&
034152 103254&
034152 104523&
034152 130524&
034152 140523\\

034152 230154&
034152 230514&
034152 240513&
034215 103254&
034215 104523&
034215 230154\\

034215 230514&
034215 235014&
034215 235104&
034251 103254&
034251 104523&
034251 230154\\

034251 230514&
034251 240153&
034251 240513&
034251 245013&
034512 103254&
034512 104523\\

034512 130524&
034512 140523&
034512 230154&
034512 230514&
034512 240513&
034512 310524\\

034512 325014&
034521 103254&
034521 104523&
034521 130524&
034521 140523&
034521 230154\\

034521 230514&
034521 240153&
034521 240513&
034521 245013&
034521 310524&
034521 325014\\
\hline
035214 143025&
035214 145203&
035214 201543&
035214 205143&
035214 210543&
035214 250143\\

035214 254013&
035214 254103&
035241 201543&
035241 205143&
035241 210543&
035241 250143\\

042315 201453&
042315 201543&
042315 205143&
042315 210543&
042315 215043&
042351 201453\\

042351 201543&
042351 204153&
042351 204513&
042351 205143&
042351 205413&
042351 210453\\

042351 210543&
042351 214053&
042351 215043&
042351 240153&
042351 240513&
042351 241053\\

042351 245013&
042351 245103&
042513 153042&
042513 210543&
042513 215043&
042531 201543\\

042531 210534&
042531 210543&
042531 215043&
043215 103254&
043215 104523&
043215 201453\\

043215 201543&
043215 204153&
043215 204513&
043215 205143&
043215 210543&
043215 240513\\

043251 103254&
043251 201453&
043251 201543&
043251 204153&
043251 204513&
043251 205143\\

043251 210453&
043251 210543&
043251 214053&
043251 214503&
043251 215043&
043251 215403\\
\hline
043251 230154&
043251 240153&
043251 240513&
043251 245013&
043251 250143&
043251 251043\\

043251 251403&
043251 254013&
043251 254103&
043512 103254&
043512 104523&
043512 140523\\

043512 204513&
043512 240513&
043512 310524&
043512 325014&
043521 104523&
043521 201453\\

043521 204153&
043521 204513&
043521 210453&
043521 214053&
043521 214503&
043521 230154\\

043521 240153&
043521 240513&
043521 245013&
043521 310524&
043521 325014&
052341 201453\\

052341 201543&
052341 204153&
052341 210453&
052341 210543&
052341 214053&
052341 214503\\

052341 215043&
053241 201453&
053241 201543&
053241 204153&
053241 205143&
053241 210453\\

053241 210543&
053241 215043&
053241 245013&
102453 210534&
103254 210543&
103254 215043\\

103254 215403&
103524 254103&
103542 325014&
104352 310524&
104523 210534&
104523 215034\\

104523 215304&
104523 251034&
104532 310524&
120534 201453&
120534 204153&
120534 204513\\
\hline
120534 240513&
120534 245013&
120543 230514&
120543 235014&
120543 253014&
125034 204153\\

125034 204513&
125034 240513&
125034 245013&
125043 230514&
125304 204153&
125304 204513\\

125304 240153&
125304 240513&
125304 245013&
130524 210543&
130524 215043&
130542 320514\\

130542 325014&
135024 210543&
135024 215043&
135204 210543&
&\\
\hline
\end{tabular}
\caption{Pairs of permutations of size $6$ that form a forbidden triple together with $012345$. These are all forbidden triples up to normalization.}
\end{center}
\end{table}
}

\end{document}